\documentclass[letterpaper, DIV=17]{scrartcl}

\usepackage{amssymb, amsthm, mathtools,bbm}
\usepackage[square,sort,comma,numbers]{natbib}
\usepackage{comment}

\usepackage[none]{hyphenat}
\usepackage[english]{babel}
\usepackage{enumerate}
\usepackage{graphicx}%
\usepackage{subcaption}
\usepackage{dsfont}%
\usepackage{float}%
\usepackage{hyperref}
\usepackage{color}
\usepackage[utf8]{inputenc}
\usepackage{graphicx}
\usepackage{verbatim} 
	\usepackage{hyperref}
	\usepackage{comment}
	\usepackage{xfrac} 
	\usepackage{tensor} 
	\usepackage{listings}  

	\usepackage{pgfplots}              
	\pgfplotsset{compat=newest}
	\usepgfplotslibrary{fillbetween}

	\usepackage[title]{appendix}

\allowdisplaybreaks
	\numberwithin{equation}{section}
\newtheorem{theorem}{Theorem}[section]
\newtheorem{corollary}[theorem]{Corollary}
\newtheorem{proposition}[theorem]{Proposition}
\newtheorem{lemma}[theorem]{Lemma}
\theoremstyle{definition}

\theoremstyle{remark}

\newcommand{\y}{\mathbf{y}}
\newcommand{\R}{\mathbb{R}}

\newcommand{\N}{\mathbb{N}}
\newcommand{\E}{\mathbb{E}}
\newcommand{\Pp}{\mathbb{P}}
\newcommand{\one}{\mathbbm{1}}
\newcommand{\Tr}{\operatorname{Tr}}
\newcommand{\Var}{\operatorname{Var}}
\newcommand{\Ent}{\operatorname{Ent}}

\newcommand{\Av}{\operatorname{Av}}
\newcommand{\Ran}{\operatorname{Ran}}
\newcommand{\dd}{\,\mathrm{d}}
\newcommand{\HS}{\mathrm{HS}}
\newcommand{\tr}{\operatorname{tr}}
\newcommand{\bxi}{\boldsymbol{\xi}}
\newcommand{\bh}{\mathbf{h}}
\newcommand{\bw}{\mathbf{w}}
\newcommand{\norm}[1]{\left\lVert #1\right\rVert}

\newcommand{\inner}[2]{\left\langle #1,#2\right\rangle}
\renewcommand{\epsilon}{\varepsilon}

\begin{document}
		\title{Correlation inequalities for transversal field models with application to quantum glasses}
		
		\author{Chokri Manai$^{1,5}$ and Simone Warzel$^{2,3,4}$ \\
			\\
            \small $^1$ Courant Institute, New York University, USA \\[-.5ex]
			\small $^2$ Department of Mathematics, TU Munich, Garching, Germany \\[-.5ex]
			\small $^3$ Munich Center for Quantum Science and Technology, Munich, Germany \\[-.5ex]
			\small $^4$ Department of Physics, TU Munich, Garching Germany,\\[-.5ex]
			\small $^5$ Corresponding author: cm7411@nyu.edu}

        \date{\small \today \\[-.5cm]}
		\maketitle		
		
		\begin{abstract}
	   We study the class of transversal field models with pair interactions, including the quantum Curie-Weiss and the transversal field Ising model, as well as quantum spin glasses based on the Sherrington-Kirkpatrick (SK) and Edwards-Anderson models. Our main result is a general correlation inequality based on a Gaussian convolution estimate, which is derived from an extension of the Ding-Song-Sun inequality and an application of Brascamp-Lieb convexity techniques. The results apply for sufficiently strong transversal fields and, for quantum spin glasses, yield a high-field regime without replica order. Combined with the previously established low-field replica-symmetry-breaking regime, this gives distinct low- and high-field overlap phases. As a corollary of the logarithmic Sobolev inequalities for the underlying functional integrals derived here, we show the concentration of the self- and replica-overlap, and determine the high-field asymptotics of the pressure in the quantum SK model. \\
       
            \noindent
          {\small \textbf{Keywords:} Brascamp-Lieb, logarithmic Sobolev inequality, phase transition, spin glass, quantum Almeida-Thouless line }
		\end{abstract}
	
\bigskip		
\tableofcontents
\bigskip

\section{Introduction to  transverse-field models}

Transverse-field models form a broad and versatile class in statistical mechanics~\cite{BSS96,SIC13}. Within the larger zoo of possible quantum extensions of classical spin systems, adding a transverse field is arguably the simplest and most natural way: it creates quantum effects through noncommuting terms, represented on each spin-$\frac{1}{2}$ by the Pauli matrices,
$$
S^x \coloneqq \left( \begin{matrix}
    0 & 1 \\ 1 & 0 
\end{matrix}\right) , \quad S^y \coloneqq \left( \begin{matrix}
    0 & -i \\ i & 0 
\end{matrix}\right), \quad S^z \coloneqq \left( \begin{matrix}
    1 & 0 \\ 0 & -1 
\end{matrix}\right) ,
$$
while preserving much of the underlying classical interaction structure. 
Lifting the above matrices to the tensor product Hilbert space $$\mathcal{H}_N \coloneqq \bigotimes_{j=1}^N \mathbb{C}^2 $$ of $ N $-spin-$\frac{1}{2}$ by letting them act trivially in the other tensor components, i.e.\ $ S_j^x \coloneqq \mathbbm{1} \otimes \dots  \mathbbm{1} \otimes S^x \otimes \mathbbm{1}  \dots \otimes  \mathbbm{1}  $ for $ j \in \{ 1, \dots , N \} $, and similarly for the $y$- and $z$-Pauli matrices, the transversal-field Hamiltonians with pair interactions are of the form 
\begin{equation}\label{eq:Ham}
 H_N({J}, \mathbf{h}) \coloneqq - \frac{1}{2}\sum_{j,k=1}^N J_{j,k} S_j^z S_k^z - \sum_{j=1}^N h_j S_j^z - b \sum_{j=1}^N S_j^x , 
\end{equation}
and act on  $ \mathcal{H}_N $.
They are characterized by:
\begin{enumerate}
    \item an interaction matrix  ${J} \coloneqq  ( J_{j,k} )_{j,k = 1 , \dots , N} \geq 0 $, which we assume here to be positive semidefinite, 
    \item longitudinal fields $ \mathbf{h} = (h_1,\dots , h_N) \in \mathbb{R}^N $, and 
    \item a constant transversal field $ b > 0 $.
\end{enumerate}
These properties are assumed throughout this paper.
This class of models encompasses the transverse-field Ising model, its mean-field cousin, the quantum Curie–Weiss model, as well as quantum spin glasses such as the models based on the classical Edwards–Anderson glass, or its mean-field cousin, the  Sherrington-Kirkpatrick (SK) glass. \\

We will be interested in the thermodynamic properties of this broad class. They are described in terms of the normalized partition function, 
\begin{equation}\label{eq:partition}
 Z_N(J,\bh)
 =\frac{\Tr e^{-\beta H_N(J,\bh)}}{(2\cosh(\beta b))^N}
\end{equation}
at inverse temperature $ \beta > 0 $. The normalization is chosen such that at $ {J} = 0 $ and $ \mathbf{h} = 0 $ the partition function is one at all temperatures. The corresponding quantum Gibbs state is abbreviated by 
\begin{equation}\label{def:Gibbs}
 \langle (\cdot) \rangle_{J,\bh} \coloneqq \frac{\Tr(e^{-\beta H_N(J,\bh)} (\cdot) )}{\Tr e^{-\beta H_N(J,\bh)}}  .
\end{equation}
The positive-semidefinite assumption on $ J $ is less restrictive than it
first appears, because diagonal entries of the interaction only shift the energy and hence leave the Gibbs measure unchanged. We will use this trick in our applications to spin glasses. 

At vanishing longitudinal field $\bh=\mathbf{0}$, the Hamiltonian~\eqref{eq:Ham} is invariant under the global spin-flip unitary $U=\prod_{j=1}^N S_j^x$, which implements rotation by $\pi$ about the $x$-axis. Hence, the $z$-magnetization of each spin $ j \in \{ 1, \dots , N \} $  vanishes, 
\begin{equation}\label{eq:nohmag}
    \langle S^z_j\rangle_{{J},\mathbf{0}} = 0 ,
\end{equation}
as in the corresponding classical Ising-type model at $ b = 0 $.  
A common feature of the above models is the emergence of a paramagnetic phase for large enough transversal field $ b $ at $ \mathbf{h} = 0 $. Remarkably, this paramagnetic phase extends down to zero temperature ($\beta \to \infty)$, where it ends in a quantum phase transition of the ground state. The main aim of this paper is to derive general bounds on the location of the paramagnetic phase as well as inequalities on the longitudinal spin correlations and their concentration. These bounds will then be applied to the quantum SK model, where we establish the absence of replica-symmetry breaking for all temperatures at high enough transversal field.

\subsection{Correlation inequalities}

Our first main result controls longitudinal susceptibilities uniformly in $\bh$.
Here and below $\norm{J}$ denotes the operator norm of the interaction matrix.

\begin{theorem}[Uniform Hessian bound]\label{thm:Hessian}
Let $J\geq0$ and assume
\begin{equation}\label{eq:high-field-condition}
 b>\norm{J}\tanh(\beta b).
\end{equation}
Then, the bound on the Hessian 
\begin{equation}\label{eq:Hessian}
 D_{\bh}^2\log Z_N(J,\bh)
 \leq
 \frac{\beta\tanh(\beta b)}{b-\norm{J}\tanh(\beta b)}\,\one
\end{equation}
holds for every $\bh\in\R^N$ in the sense of quadratic forms.
\end{theorem}
The proof is found in Subsection~\ref{sub:Hessian}. 
To spell out the consequence, we recall \cite{Simon_2025} the Duhamel-Kubo
inner product of two operators $ A, B $ on $\mathcal{H}_N $: 
\begin{equation}\label{eq:Duhamel-def}
 (A,B)_{J,\bh}
 :=\frac{1}{\Tr e^{-\beta H_N}}
 \int_0^1\Tr\!\left(e^{-(1-s)\beta H_N}A^*e^{-s\beta H_N}B\right)\dd s,
\end{equation}
where $H_N \equiv H_N(J,\bh)$.  The Duhamel-Kubo inner product is related to the second derivatives of the partition function. In particular for $B_{\mathbf v}=\sum_jv_jS_j^z$ and
$\widetilde B_{\mathbf v} \coloneqq B_{\mathbf v}-\langle B_{\mathbf v}\rangle_{J,\bh}$, differentiation gives
\begin{equation}\label{eq:Hessian-Duhamel}
 \inner{\mathbf v}{D_{\bh}^2\log Z_N(J,\bh)\mathbf v}
 =\beta^2(\widetilde B_{\mathbf v},\widetilde B_{\mathbf v})_{J,\bh}.
\end{equation}
For a self-adjoint observable $A$, the Falk--Bruch inequality
\cite{FB69,DLS78} in the convenient form
\[
 \langle A^2\rangle
 \leq (A,A) +
 \frac12\sqrt{(A,A) \,
 \langle[A,[\beta H_N,A]]\rangle}
\]
and the identity
\[
 [B_{\mathbf v},[\beta H_N,B_{\mathbf v}]]
 =4\beta b\sum_{j=1}^Nv_j^2S_j^x
\]
together with Theorem~\ref{thm:Hessian} and the bound $\langle S_j^x\rangle\leq1$ imply
\begin{align}
 \big\langle\widetilde B_{\mathbf v}^{\,2}\big\rangle_{J,\bh}
 \leq & \mathfrak m(\beta,b;\norm{J})\norm{\mathbf v}_2^2 \notag \\
&  \quad \mbox{with} \;   \mathfrak m(\beta,b;\norm{J}) :=m(\beta,b;\norm{J})+\sqrt{\beta b\,m(\beta,b;\norm{J})}, \label{def:mathfrakm} \\
&  \quad \mbox{and} \; m(\beta,b;\norm{J}) \coloneqq \frac{\tanh(\beta b)}{\beta\bigl(b-\norm{J} \tanh(\beta b)\bigr)} . \label{def:m}
\end{align}
In particular, if \eqref{eq:high-field-condition} applies, we conclude that 
\begin{equation}\label{eq:magnetization-variance}
 \Bigg\langle\Big(\frac1N\sum_{j=1}^N
 (S_j^z-\langle S_j^z\rangle_{J,\bh})\Big)^2\Bigg\rangle_{J,\bh}
 \leq\frac{\mathfrak m(\beta,b;\norm{J})}{N}.
\end{equation}
At $ \mathbf{h} = 0 $ where~\eqref{eq:nohmag} holds, this establishes the vanishing of the specific $ z$-magnetization in the thermodynamic limit $ N \to \infty $. 
More generally, the correlation matrix $ C(J,\bh) $ defined by 
\begin{equation}\label{def:cormat}
 C_{jk}(J,\bh)
 :=\left\langle
 (S_j^z-\langle S_j^z\rangle)
 (S_k^z-\langle S_k^z\rangle)
 \right\rangle_{J,\bh}, 
\end{equation}
which is real and positive semidefinite, satisfies the following norm bounds. 

\begin{corollary}[Correlation-matrix bound]\label{cor:correlation-matrix}
If \eqref{eq:high-field-condition} holds, then uniformly in $\bh$,
\begin{equation}\label{eq:Cor}
 \norm{C(J,\bh)}\leq\mathfrak m(\beta,b;\norm{J}),
 \qquad
 \frac1N\sum_{j,k=1}^N C_{jk}(J,\bh)^2
 \leq\mathfrak m(\beta,b;\norm{J})^2.
\end{equation}
\end{corollary}

\begin{proof}
For every $\mathbf v$, the quadratic form $\mathbf v^T C\mathbf v$ is the variance appearing in the estimate preceding
\eqref{def:mathfrakm}, such that $\norm{C}\leq\mathfrak m$.  Since $C\geq0$,
$\norm{C}_{\HS}^2\leq N\norm{C}^2$.
\end{proof}

The above results in particular apply to the following popular examples.
\begin{description}
\item[Quantum Curie--Weiss model]
For $J_{jk}=N^{-1}$, one has $J\geq0$ and $\norm{J}=1$.  At zero longitudinal
field, the condition $b>\tanh(\beta b)$ agrees with the exact paramagnetic
criterion for the quantum Curie--Weiss model \cite{Chayes:2008aa} (see also~\cite{MW23b,Manai25}).
This demonstrates that  the condition~\eqref{eq:high-field-condition} is optimal in this generality. 
\item[Transverse-field Ising model]\label{ex:Ising}
Let $A$ be the adjacency matrix of a finite subgraph of $\mathbb Z^d$.  Since
$\norm{A}\leq2d$, the shifted matrix $J=A+2d\one$ is positive semidefinite with
$\norm{J}\leq4d$.  It defines the same Gibbs
measure as $A$.  Hence all bounds above hold whenever
$b>4d\tanh(\beta b)$, in particular for every temperature if $b>4d$.
This slightly improves the transversal field regime in which stronger spatial Ornstein--Zernike estimates have been derived in~\cite{Kenn91}.
\end{description}
Applications to quantum spin glasses will be discussed in Section~\ref{sec:applications}. 

\subsection{Functional integral representation}
Our proof of Theorem~\ref{thm:Hessian} is based on the functional 
integral representation of the normalized partition function~\eqref{eq:partition},
and the corresponding Gibbs state~\eqref{def:Gibbs}, which we discuss next.

Let $\nu_1\equiv\nu_1^{\beta b}$ be the following probability measure on
c\`adl\`ag paths $\xi:[0,1]\to\{-1,1\}$ which are periodic, $ \xi(0) = \xi(1) $.  We choose $\sigma\in\{-1,1\}$ uniformly and independently sample a Poisson point process $\omega$ on $[0,1]$ with intensity $\beta b$, conditioned to contain an even number of points. The path then arises by flipping the sign of $\xi$ at every Poisson point, $\xi(s)\coloneqq\sigma(-1)^{\omega([0,s])}$.  The even-parity conditioning makes the path
periodic.  For the partition function, we need the $ N $-fold product $\nu_N=\nu_1^{\otimes N}$.  The standard Trotter/Feynman--Kac
representation (see e.g.\ \cite[App.~B]{Leschke:2021aa}), gives
\begin{align}
 Z_N(J,\bh)
 &=\int \exp\left\{
 \frac\beta2\int_0^1\sum_{j,k=1}^NJ_{jk}\xi_j(s)\xi_k(s)\dd s
 +\beta\int_0^1\sum_{j=1}^Nh_j\xi_j(s)\dd s
 \right\}\nu_N(\dd\bxi), \notag \\
 \langle (\cdot) \rangle_{J,\bh}
 &=\frac{1}{Z_N(J,\bh)}\int (\cdot) \exp\left\{
 \frac\beta2\int_0^1\sum_{j,k=1}^NJ_{jk}\xi_j(s)\xi_k(s)\dd s
 +\beta\int_0^1\sum_{j=1}^Nh_j\xi_j(s)\dd s
 \right\} \nu_N(\dd\bxi).
 \label{eq:path-Gibbs}
\end{align}
For one spin and zero longitudinal field, the correlation function is known explicitly (cf.~\cite{Leschke:2021aa})
\begin{equation}\label{eq:free-loop-covariance}
 \int\xi(t)\xi(s)\nu_1(\dd\xi)
 =\frac{\cosh\!\bigl(\beta b(1-2|t-s|)\bigr)}{\cosh(\beta b)} =: \mu_0(t,s) .
\end{equation}
This kernel defines a non-negative Hilbert-Schmidt operator on the Hilbert space $ L^2(0,1) $.

\subsection{Self-overlap}

The path self-overlap in the above functional integral is the Hilbert--Schmidt kernel
\begin{equation}\label{def:so}
 S_N(t,s)=\frac1N\sum_{j=1}^N\xi_j(t)\xi_j(s)
 .
\end{equation}
It defines a positive operator $ S_N  = \frac1N\sum_{j=1}^N|\xi_j\rangle\langle\xi_j| $ of rank at most $N$ on  $L^2(0,1)$, and its Hilbert-Schmidt norm is bounded, 
$\norm{S_N}_{\HS}^2 \coloneqq \int_0^1\int_0^1  S_N(t,s)^2 dt ds  \leq1$. In the high-field regime, the self-overlap concentrates on its Gibbs average. 
\begin{theorem}[Self-overlap concentration]\label{thm:concentration1}
Let $J\geq0$ and suppose \eqref{eq:high-field-condition} holds.
Then, for every $\bh\in\R^N$,
\begin{equation}\label{eq:self-variance}
 \left\langle\norm{S_N-\langle S_N\rangle_{J,\bh}}_{\HS}^2\right\rangle_{J,\bh}
 \leq\frac{1+r(\beta,b;\norm{J})}{N},
\end{equation}
and, for every $\epsilon>0$,
\begin{equation}\label{eq:soconc}
 \left\langle\one\!\left\{
 \norm{S_N-\langle S_N\rangle_{J,\bh}}_{\HS}
 >\epsilon+\sqrt{\frac{r(\beta,b;\norm{J})}{N}}
 \right\}\right\rangle_{J,\bh}
 \leq
 2\exp\left[-\frac{N\epsilon\min\{\epsilon,\sigma_{\mathrm s}^2\}}
 {4\sigma_{\mathrm s}^2}\right], 
\end{equation}
where
$\displaystyle \; 
 \sigma_{\mathrm s}^2=e+r(\beta,b;\norm{J}) , \quad r(\beta,b;\norm{J}) \coloneqq \frac{\norm{J}\tanh(\beta b)}{b-\norm{J}\tanh(\beta b)} $. \\
At zero longitudinal field, we also have
\begin{equation}\label{eq:sharp-self-bound}
 \norm{\langle S_N\rangle_{J,0}}_{\HS}^2
 \leq\frac12m(\beta,b;\norm{J})+m(\beta,b;\norm{J})^2.
\end{equation}
\end{theorem}
The proof is found in Subsection~\ref{sec:self-proof}.\\

For fixed $\beta$ and $\norm{J}$, the right-hand side of
\eqref{eq:sharp-self-bound} is $(2\beta b)^{-1}+O(b^{-2})$ as $b\to\infty$.
The leading constant is optimal already for one free spin since by
\eqref{eq:free-loop-covariance},
\begin{equation}\label{eq:one-spin-HS}
 \norm{\langle S_1\rangle_{0,0}}_{\HS}^2
 = \norm{\mu_0}_{\HS}^2
 =\frac{\tanh(\beta b)}{2\beta b}
  +\frac{1}{2\cosh^2(\beta b)}
 =\frac{1}{2\beta b}+o(b^{-1})
\end{equation}
where the remainder $ o(b^{-1})$ goes to zero as $ b \to \infty $.
\subsection{Replica overlap and its centered version}
In the theory of spin glasses, the second important kernel is the replica overlap, 
\begin{equation}
    R_N^{(1,2)}(t,s) \coloneqq \frac1N\sum_{j=1}^N \xi_j^{(1)}(t)\xi_j^{(2)}(s) . 
\end{equation}
It involves two copies, $(\xi_j^{(1)}, \xi_j^{(2)}) $, of the Gibbs average, and we will denote by $ \langle (\cdot) \rangle_{J,\mathbf{h}}^\otimes $ the duplicated (normalized) Gibbs average. Related, yet different, is the centered version of the replica overlap
\begin{equation}\label{def:replica}
 C_N^{(1,2)}(t,s)
 =\frac1N\sum_{j=1}^N
 \bigl(\xi_j^{(1)}(t)-\bar\xi_j\bigr)
 \bigl(\xi_j^{(2)}(s)-\bar\xi_j\bigr), \qquad  \bar\xi_j \coloneqq \langle\xi_j(t)\rangle_{J,\bh}=\langle S_j^z\rangle_{J,\bh} .
\end{equation}
It depends on $ J, \bh $, which we suppress in the notation, and the centering $\bar\xi_j  $ is independent of time due to the time-invariance of the Gibbs measure.
Under the duplicated Gibbs measure, the average of the Hilbert-Schmidt norm of the centered replica overlap $  C_N^{(1,2)} $ is related to the correlation matrix defined in~\eqref{def:cormat}. In fact, for any $ t, s \in [0,1) $ one has:
\begin{align}\label{eq:orderp}
	 \big\langle C_N^{(1,2)}(t,s)^2\big\rangle_{J,\mathbf{h}}^\otimes  & =   \frac{1}{N^2} \sum_{j,k=1}^N    \left\langle \left( S_j^z(t) - \langle S_j^z(t)\rangle_{J,\mathbf{h}} \right) \left( S_k^z(t) - \langle S_k^z(t)\rangle_{J,\mathbf{h}}\right)\rangle_{J,\mathbf{h}} \right.  \notag \\
    & \mkern200mu \times \left\langle \left( S_j^z(s) - \langle S_j^z(s)\rangle_{J,\mathbf{h}} \right) \left( S_k^z(s) - \langle S_k^z(s)\rangle_{J,\mathbf{h}}\right)\right\rangle_{J,\mathbf{h}}  
   \notag \\ 
	& =    \frac{1}{N^2} \sum_{j,k=1}^N    \langle \left( S_j^z - \langle S_j^z \rangle_{J,\mathbf{h}}\right) \left( S_k^z -  \langle S_k^z \rangle_{J,\mathbf{h}}\right) \rangle_{J,\mathbf{h}}^2 = \frac{1}{N^2} \sum_{j,k=1}^N  C_{jk}(J,\mathbf{h})^2 .
\end{align}
In the first step, we translated the functional integral expectations to Gibbs expectations of operators $ S_j^z(t) \coloneqq e^{t \beta H_N}S_j^z e^{-t \beta H_N} $. 
The second step follows from the time-translation invariance of the equal-time Gibbs correlations.

At $\bh=0$, the centered overlap \eqref{def:replica} coincides with the usual
replica overlap due to the global spin-flip symmetry~\eqref{eq:nohmag}. More generally, we have 
$$
R_N^{(1,2)}(t,s) - \mathfrak{r}_N^2 =  C_N^{(1,2)}(t,s) + \frac{1}{N} \sum_{j=1}^N \left(\xi_j^{(1)}(t) - \bar\xi_j \right) \bar\xi_j + \frac{1}{N} \sum_{j=1}^N \left(\xi_j^{(2)}(s) - \bar\xi_j \right) \bar\xi_j ,
$$
with $ \mathfrak{r}_N^2 \coloneqq \langle R_N^{(1,2)}(t,s) \rangle_{J,\bh}^\otimes = \frac{1}{N} \sum_{j=1}^N (\bar\xi_j)^2 $ 
such that 
\begin{align}
  \left\langle \norm{R_N^{(1,2)}- \mathfrak{r}_N^2}_{\HS}^2 \right\rangle_{J,\bh}^\otimes
  &=\frac{1}{N^2} \sum_{j,k=1}^N \bigl( C_{jk}(J,\mathbf{h})^2  +2\bar\xi_j \ \bar\xi_k \ C_{jk}(J,\mathbf{h}) \bigr)
  \label{eq:quenched-variance}
\end{align}

 If $ b > \| J \| \tanh(\beta b) $, Corollary~\ref{cor:correlation-matrix} shows that 
$ \big\langle \big\| C_N^{(1,2)} \big\|_\HS^2\big\rangle_{J,\mathbf{h}}^\otimes \leq \mathfrak{m}(\beta,b;\norm{J})^2 \ N^{-1} $, 
 i.e., the centered replica overlap vanishes in the thermodynamic limit, and so does $ \big\langle \big\| R_N^{(1,2)}- \mathfrak{r}_N^2\big\|_{\HS}^2 \big\rangle_{J,\bh}^{\otimes}  $.  In case $ \bh = 0 $ where~\eqref{eq:nohmag} holds, one has $\bar\xi_j=0$ for all $j$, and hence $\mathfrak r_N^2=0$. By~\eqref{eq:orderp} and~\eqref{eq:quenched-variance},  this then implies that the non-centered replica overlap vanishes as well in this regime. 
 
Our next result shows that the centered replica overlap concentrates under the duplicated Gibbs measure.

\begin{theorem}[Replica-overlap concentration]\label{thm:repconc}
Under the assumptions of Theorem~\ref{thm:concentration1}, set
$
 \sigma_{\mathrm r}^2=16e+8 \ r(\beta,b;\norm{J}) $. 
Then, for every $\epsilon>0$,
\begin{equation}\label{eq:reconc}
 \left\langle\one\!\left\{
 \norm{C_N^{(1,2)}}_{\HS}
 >\epsilon+\frac{\mathfrak m(\beta,b;\norm{J})}{\sqrt N}
 \right\}\right\rangle_{J,\bh}^{\otimes}
 \leq
 2\exp\left[-\frac{N}{16}
 \min\left\{\epsilon,\frac{4\epsilon^2}{\sigma_{\mathrm r}^2}\right\}\right].
\end{equation}
\end{theorem}
The proof is found in Subsection~\ref{sec:replica-proof}. 
The bound~\eqref{eq:reconc} in particular implies that all moments of the centered replica-overlap vanish as $N \to \infty$, i.e., for any $ p \in \mathbb{N} $ there is some $ C_p \equiv C_p(\beta , b; \| J \|) $ such that for all $ N $:
\begin{equation}\label{eq:replmom}
    \big\langle \big\| C_N^{(1,2)} \big\|^p_\HS\big\rangle_{J,\mathbf{h}}^\otimes \leq \frac{C_p(\beta,b;\norm{J})}{N^{p/2}} .
\end{equation}

\section{Applications to quantum spin glasses}\label{sec:applications}

\subsection{Quantum Sherrington--Kirkpatrick model}

Let $g_{jk}$, $1\leq j,k\leq N$, be independent standard Gaussian variables and
abbreviate by
\begin{equation}\label{def:GOE}
 G_{j,k} =\frac{g_{jk}+g_{kj}}{\sqrt{2N}}
\end{equation}
the matrix elements of a standard GOE random matrix.  The quantum
Sherrington--Kirkpatrick (QSK) Hamiltonian is the random transversal field model
\begin{equation}\label{def:QSK}
 H_N(G,\bh)
 =-\frac12\sum_{j,k=1}^N G_{jk}S_j^zS_k^z
  -\sum_{j=1}^Nh_jS_j^z-b\sum_{j=1}^NS_j^x.
\end{equation}
This model and its transverse-field phase diagram have been studied extensively in the physics literature~\cite{FS86,Usadel:1987aa,Yamamoto:1987aa,RCC89,GL90,Young:2017aa,TTC17}, but much less is known with mathematical rigor~\cite{AdBr20,Leschke:2021aa,Leschke:2021ab,Itoi:2023aa,MW25}.\\

In the classical model, replica-symmetry breaking is encoded by Parisi's overlap
order parameter \cite{Parisi:1980aa,Mezard:1986aa}. At $ b = 0 $ and $ \bh = 0 $ the second moment of the replica overlap is $ \mathbb{E}\left[ \langle\| C_N^{(1,2)} \|^2_{\HS} \rangle_{G,\mathbf{0}}^\otimes \right] $, cf.~\eqref{eq:orderp} and~\eqref{eq:quenched-variance}. At constant non-vanishing longitudinal field, $ \bh = (h, \dots, h) $, the classical Almeida--Thouless (AT)
criterion describes the stability of the replica-symmetric solution in a constant 
longitudinal field~\cite{AT78}; see~\cite{Toninelli02,L26} for a proof. The AT criterion guarantees that at any $ h \neq 0 $,  there is a spin-glass phase with non-zero Parisi order parameter at sufficiently large $ \beta $. 
In contrast, replacing the longitudinal field by a transversal one, the spin-glass phase is predicted to cease to exist at sufficiently large $ b> 0 $ at any temperature. The critical $ b_c $ at which this happens at zero temperature marks a quantum phase transition.
There is, however, no closed-form prediction for the quantum Almeida-Thouless line in the QSK, i.e.\ $ \bh = 0 $ and $ b> 0 $ -- only for its related self-overlap corrected cousin~\cite{MW26}, and simpler quantum glass models \cite{MW21, MW20, MW22, KMW25} and related trajectory dynamics \cite{GMW22, MW25b}. The Parisi description for the QSK~\cite{MW25} is substantially more involved than in the classical case, and, aside from the replica order parameter, the self-overlap is a central quantity -- as it is in the closely related vector spin glasses~\cite{Chen23,Chen24}. \\ 

Our first main result for the QSK is the absence of replica order at sufficiently large transverse field. Combined with the low-field RSB result of~\cite{Leschke:2021ab}, it establishes a transition between low- and high-field overlap regimes. 
\begin{theorem}[High-field absence of replica order]\label{thm:norsb}
Assume
\begin{equation}\label{eq:QSK-condition}
 b>4\tanh(\beta b), 
\end{equation}
and let $\delta>0$ be small enough so that
$b>(4+2\delta)\tanh(\beta b)$. 
Then, on an event $ \Omega_{N,\delta} $, whose probability is exponentially close to one for large $N$,  and at every $\bh\in\R^N$,
\begin{equation}\label{eq:QSK-noRSB}
 \left\langle\norm{C_N^{(1,2)}}_{\HS}^2\right\rangle_{G,\bh}^{\otimes}
 \leq\frac{\mathfrak m(\beta,b;4+2\delta)^2}{N}.
\end{equation}
In particular, at $ \bh = 0 $ the replica order parameter vanishes in the
thermodynamic limit.
\end{theorem}
\begin{proof}
On the event 
\begin{equation}\label{def:omega}
    \Omega_{N,\delta} \coloneqq \{\norm{G}\leq2+\delta\} , 
\end{equation} 
the interaction matrix $
 J \coloneqq G+(2+\delta)\one\geq0 $ satisfies 
$ \norm{J}\leq4+2\delta $. 
Since the interaction matrices $J$ and $G$ define the same Gibbs state, 
\eqref{eq:QSK-noRSB} follows from \eqref{eq:orderp} and \eqref{eq:Cor}.  The Bai--Yin law
\cite{BY88} implies that $\Omega_{N,\delta}$ has probability close to one for all
sufficiently large $N$ -- by standard Gaussian concentration of measure~\cite{Ledoux01}, this probability is exponentially close to one. 
\end{proof}

The criterion~\eqref{eq:QSK-condition} is uniform in temperature and, in particular, holds for all
$\beta \geq 0 $ when $b>4$. 
It hence also includes the large $ \beta  $ part of the regime 
$$
\max_{q \in [0,1]} \left(\frac{\beta^2}{4} \left( 1 - (1-q)^2 \right) - \mathbb{E} \log \left[ \cosh\left(\beta g \sqrt{q} \right)\right] \right)   > \log \cosh(\beta b) ,
$$
with an  $ N(0,1) $-random variable $ g $, whose expectation is denoted by $ \mathbb{E} $.
For this regime and at $ \mathbf{h}= \mathbf{0}$, it was shown in \cite[Eq.~(3.15)]{Leschke:2021aa} that the annealed free energy deviates from the quenched free energy. Consequently,  as in the SK model with a longitudinal field above the AT-line, the (unconditioned) second-moment method does not yield a proof of the absence of RSB in this regime.

 The criterion~\eqref{eq:QSK-condition} is not expected to be sharp. Technically, it shares the same weakness as similar conditions in a simple proof of rapid mixing of the Glauber dynamics in the classical SK model~\cite{BaBo19}. For our question, weak-disorder methods give \cite{Leschke:2021aa} 
a complementary replica-symmetric region, while replica
symmetry breaking is known to persist at low temperature and sufficiently small
transverse field \cite{Leschke:2021ab}.  Numerical and perturbative studies
place the zero-temperature transition $ b_c \approx 1.56\dots $, which is below the present sufficient threshold
\cite{Yamamoto:1987aa,Young:2017aa}.

\subsection{Sample-dependent concentration and the pressure at large field}

The following is the direct disorder-averaged consequence of
Theorem~\ref{thm:concentration1}.  

\begin{theorem}[QSK self-overlap concentration]\label{thm:concSK}
Under \eqref{eq:QSK-condition}, choose $\delta$  as in
Theorem~\ref{thm:norsb}, and set
$r_\delta\coloneqq r(\beta,b;4+2\delta)$ and $\sigma_\delta^2=e+r_\delta$.  For every $\epsilon>0$,
\begin{align}
 &\E\left\langle\one\!\left\{
 \norm{S_N-\langle S_N\rangle_{G,\bh}}_{\HS}
 >\epsilon+\sqrt{\frac{r_\delta}{N}}
 \right\}\right\rangle_{G,\bh} \leq
 2\exp\left[-\frac{N\epsilon\min\{\epsilon,\sigma_\delta^2\}}
 {4\sigma_\delta^2}\right]
 +\Pp(\Omega_{N,\delta}^{\mathrm c}).
 \label{eq:SKconcSo}
\end{align}
At $\bh=0$ and on $\Omega_{N,\delta}$, one also has
\begin{equation}\label{eq:QSK-self-mean}
 \norm{\langle S_N\rangle_{G,0}}_{\HS}^2
 \leq\frac12m(\beta,b;4+2\delta)+m(\beta,b;4+2\delta)^2 .
\end{equation}
\end{theorem}

\begin{proof}
On $\Omega_{N,\delta}$ we use the shifted matrix
$J=G+(2+\delta)\one$ and apply
Theorem~\ref{thm:concentration1}.  On the complement, we bound the probability by
one and average over the disorder. The last inequality is~\eqref{eq:sharp-self-bound}. 
\end{proof}

Estimate \eqref{eq:SKconcSo} centers at the
sample-dependent Gibbs mean $\langle S_N\rangle_{G,\bh}$. It does not assert that this random center concentrates around its disorder average.  Such
a deterministic-centering statement requires a separate disorder
self-averaging estimate. At $ \bh = \mathbf{0} $, this follows from standard Gaussian concentration of measure combined with Theorem~\ref{thm:concSK} and~\ref{thm:repconc}. We spell this out in Appendix~\ref{app:sac}.\\ 

At zero longitudinal field $ \bh = \mathbf{0} $ one defines the normalized quenched pressure
\begin{equation}\label{def:QSK-pressure}
 p_N(\beta,b)
 :=\frac1N\E\log\frac{\Tr e^{-\beta H_N(G,\mathbf{0})}}
 {(2\cosh(\beta b))^N}.
\end{equation}
Its limit $ N \to \infty $ exists and is expressed in terms of a quantum Parisi formula~\cite{MW25}. In the following, we determine its asymptotic behavior as $ b \to \infty$. 
\begin{corollary}[Large-field pressure bound]\label{cor:pressure-bound}
If \eqref{eq:QSK-condition} holds and $ \bh = \mathbf{0} $, then
\begin{equation}\label{eq:pressure-bound}
\frac{\beta^2}{8} \left( \frac{\tanh(\beta b)}{\beta b } + 1 - [\tanh(\beta b)]^2 \right) \leq \lim_{N\to\infty}p_N(\beta,b) 
 \leq\frac{\beta^2}{4}
 \left(\frac12m(\beta,b;4)+m(\beta,b;4)^2\right).
\end{equation}
Moreover, for fixed $\beta$ and as $ b\to\infty $:
\begin{equation}\label{eq:pressure-asymptotic}
 \lim_{N\to\infty}p_N(\beta,b)
 = \frac{\beta}{8b}+O_\beta(b^{-2}).
\end{equation}
\end{corollary}
\begin{proof}
Replacing $G$ by $\sqrt t\,G$ and denoting the corresponding pressure by
$p_N(t)$,  Gaussian integration by parts gives
\begin{equation}\label{eq:pressure-interpolation}
 p_N'(t)=\frac{\beta^2}{4}\E\left\langle
 \norm{S_N}_{\HS}^2-\big\| R_N^{(1,2)}\big\|_{\HS}^2
 \right\rangle_t^{\otimes}\geq0.
\end{equation}
The non-negativity is immediate by writing the difference as a sum of
variances of the path inner products $\inner{\xi_j}{\xi_k}$.  Hence $p_N(1)\geq
p_N(0)=0$ and
\[
 p_N(1)\leq\frac{\beta^2}{4}\int_0^1
 \E\langle\norm{S_N}_{\HS}^2\rangle_t\dd t.
\]
On $\Omega_{N,\delta}$, the Gibbs state at every $t\in[0,1]$ can be represented
with the positive matrix $\sqrt t\,G+(2+\delta)\one$, whose norm is at most
$4+2\delta$.  By \eqref{eq:sharp-self-bound}, \eqref{eq:self-variance}, and the
Hilbert-space variance decomposition,
\[
 \langle\norm{S_N}_{\HS}^2\rangle_t
 \leq\frac12m(\beta,b;4+2\delta)+m(\beta,b;4+2\delta)^2
 +\frac{1+r(\beta,b;4+2\delta)}{N}.
\]
On $\Omega_{N,\delta}^{\mathrm c}$, we use $\norm{S_N}_{\HS}\leq1$.  Letting
$N\to\infty$ and then $\delta\downarrow0$ proves
the upper bound in~\eqref{eq:pressure-bound}. 

For a lower bound, we restrict the path measure $ \nu_N $ to the set $ \mathcal{Q}_{N,\varepsilon} \coloneqq \left\{ \norm{ S_N - \mu_0 }_{\HS} \leq \varepsilon \right\} $ where the self-overlap is constrained to the free self-overlap from~\eqref{eq:free-loop-covariance}. It has been shown in~\cite[Thm.~2.6]{MW26} that in case $ b > \tanh(\beta b) $ the self-overlap-constrained pressure, which is a lower bound to $ p_N(\beta, b) $, converges as $ N \to \infty$ and subsequently $ \varepsilon \downarrow 0 $ to 
$$
\frac{\beta^2}{8} \left( \frac{\tanh(\beta b)}{\beta b } + 1 - [\tanh(\beta b)]^2 \right) = \frac{\beta^2}{4} \norm{\mu_0}_\HS^2 \leq \lim_{N\to \infty} p_N(\beta, b) .
$$
The first identity is by explicit integration. 

The proof of~\eqref{eq:pressure-asymptotic} is by explicit asymptotic evaluation of the bounds. 
\end{proof}

\subsection{Quantum Edwards--Anderson models}

Let $\Lambda$ be a finite graph with $N=|\Lambda|$, and let $K$ be any real
symmetric interaction matrix supported on its edges.  The corresponding
quantum Edwards--Anderson Hamiltonian is \eqref{eq:Ham} with $J$ replaced by a random matrix
$K$.  For every realization, its spectral diameter
\begin{equation}\label{def:spectral-diameter}
 d(K)=\lambda_{\max}(K)-\lambda_{\min}(K)
\end{equation}
measures the distance of the maximal and minimal eigenvalue of $ K $. 

\begin{corollary}[High-field Edwards--Anderson model]\label{cor:EA}
If
\begin{equation}\label{eq:EA-condition}
 b>d(K)\tanh(\beta b),
\end{equation}
then all conclusions of Theorems~\ref{thm:Hessian},
\ref{thm:concentration1}, and \ref{thm:repconc} hold with $ \norm{J} $ replaced by $ d(K)$.  In
particular,
\begin{equation}\label{eq:EA-overlap}
 \left\langle\norm{C_N^{(1,2)}}_{\HS}^2\right\rangle_{K,\bh}^{\otimes}
 \leq\frac{\mathfrak m(\beta,b;d(K))^2}{N}.
\end{equation}
If the graph has maximal degree $\Delta$ and
$|K_{xy}|\leq J_*$ on every edge $xy$, then $d(K)\leq2\Delta J_*$; hence the uniform
sufficient condition
\begin{equation}\label{eq:EA-bounded-condition}
 b>2\Delta J_*\tanh(\beta b)
\end{equation}
applies to every realization of bounded random couplings.  At $\bh=0$, the
centered overlap is the usual Edwards--Anderson replica overlap.
\end{corollary}

\begin{proof}
We shift the interaction matrix 
$J=K-\lambda_{\min}(K)\one$, so that  $J\geq0$ and $\norm{J}=d(K)$.  The bound
$d(K)\leq2\Delta J_*$ follows from $\norm{K}\leq\Delta J_*$.
\end{proof}

\section{Methods and proofs}
This section contains the core of the techniques developed in this paper. It is dedicated to the proofs of Theorem~\ref{thm:Hessian} and the concentration results, Theorems~\ref{thm:concentration1} and~\ref{thm:repconc}. Theorem~\ref{thm:Hessian} has two ingredients: 1.~a straightforward generalization of the correlation inequality of Ding-Song-Sun~\cite{DSS23} for ferromagnetic lattice systems to a quantum spin in a combination of a transversal and longitudinal field, and 2.~a Gaussian convolution inequality, which is implied by a Hubbard-Stratonovich linearization of the quadratic interaction in the functional integral combined with the inequality of Brascamp-Lieb~\cite{BL76}. 
This will also yield a logarithmic Sobolev inequality, which is one of the main ingredients in the proofs of Theorems~\ref{thm:concentration1} and~\ref{thm:repconc}. 
We start with the proof of the Ding-Song-Sun inequality adapted to our present set-up.

\subsection{A one-spin correlation inequality}

For a real $h\in L^2(0,1)$ we define
\begin{equation}\label{def:Z1h}
 Z_1(h)=\int e^{\beta\inner{h}{\xi}}\nu_1(\dd\xi),
 \qquad
 \langle F\rangle_h=Z_1(h)^{-1}\int F(\xi)e^{\beta\inner{h}{\xi}}\nu_1(\dd\xi).
\end{equation}
Let $\mu_h$ be the covariance operator on $L^2(0,1)$ with kernel
\begin{equation}\label{def:muh}
 \mu_h(t,s)=\langle\xi(t)\xi(s)\rangle_h
 -\langle\xi(t)\rangle_h\langle\xi(s)\rangle_h.
\end{equation}
At $ h = 0 $, this specializes to $ \mu_0(t,s) $ given by~\eqref{eq:free-loop-covariance}. 
\begin{proposition}[Ding--Song--Sun correlation inequality]\label{prop:DSS}
For all real $h,\varphi\in L^2(0,1)$,
\begin{equation}\label{eq:DSS}
 0\leq\inner{\varphi}{\mu_h\varphi}
 \leq\inner{|\varphi|}{\mu_0|\varphi|}
 \leq \frac{\tanh(\beta b)}{\beta b} \norm{\varphi}_2^2.
\end{equation}
\end{proposition}

\begin{proof}
It suffices by an approximation argument to take bounded piecewise-continuous $h$ and
$\varphi$.  The Trotter approximation in Appendix~\ref{app:Trotter} identifies
$\mu_h(t,s)$ as a limit of connected correlations in a ferromagnetic periodic
Ising chain.  FKG gives non-negativity, and the correlation inequality of
Ding--Song--Sun \cite{DSS23} gives the pointwise domination
$0\leq\mu_h(t,s)\leq\mu_0(t,s)$ -- hence the first two inequalities in
\eqref{eq:DSS}.  Finally, the Schur test yields
\[
 \norm{\mu_0}
 \leq\sup_s\int_0^1\mu_0(t,s)\dd t
 =\frac{\tanh(\beta b)}{\beta b}.
\]
The last step is by
\eqref{eq:free-loop-covariance} and an explicit computation of the integral. 
\end{proof}

\subsection{Gaussian convolution: Brascamp--Lieb and logarithmic Sobolev bounds}
Our next concern will be a general Gaussian convolution bound. 
Let $\Sigma>0$ be a covariance matrix in $\R^{d\times d}$ and let $F\in C^2(\R^d)$.  For
$\bh\in\R^d$ we consider the potential $  V_{\bh}: \R^d \to \R $, 
\begin{equation}\label{def:Vh}
 V_{\bh}(\y)=\frac12\inner{\y-\bh}{\Sigma^{-1}(\y-\bh)}-F(\y).
\end{equation}
If $D^2F\leq\gamma\one$ and $\gamma\norm{C}<1$, then the potential satisfies the Bakry-Emery convexity condition~\cite{BM85},
\begin{equation}\label{eq:strong-convexity}
 D^2V_{\bh}\geq \Sigma^{-1}-\gamma\one\geq\kappa\one,
 \qquad \mbox{with}\; \kappa \coloneqq \norm{\Sigma}^{-1}-\gamma>0.
\end{equation}
Brascamp and Lieb consider probability measures   $\mathcal E_{\bh}$ with density proportional to
$e^{-V_{\bh}(\y)}\dd \y$.  Their inequality \cite{BL76} asserts that 
\begin{equation}\label{eq:BL}
 \Var_{\mathcal E_{\bh}}(f)
 \leq\mathcal E_{\bh}\inner{\nabla f}{(D^2V_{\bh})^{-1}\nabla f}
 \leq\mathcal E_{\bh}\inner{\nabla f}{(\Sigma^{-1}-\gamma\one)^{-1}\nabla f}
\end{equation}
for all smooth functions $ f : \mathbb{R}^d \to \mathbb{R} $. The following is a consequence of such convexity considerations. 

\begin{lemma}[Gaussian convolution bound]\label{lem:gaussian-convolution}
Let $X$ be a
centered Gaussian vector in $ \mathbb{R}^d $ with covariance matrix $\Sigma\geq0$, and let
$F:\mathbb{R}^d \to \R$ be twice continuously differentiable with $D^2F\leq\gamma\one$ and 
$\gamma\norm \Sigma<1$.  Then $ G : \mathbb{R}^d \to \mathbb{R} $,  $
 G(h) \coloneqq \log\E \exp\left(F(X+h)\right) $
satisfies
\begin{equation}\label{eq:convolution-Hessian}
 D^2G(h)\leq\gamma(\one-\gamma \Sigma)^{-1}
\end{equation}
in the sense of quadratic forms.
\end{lemma}

\begin{proof}
Assume first $\Sigma>0$.  After the change of variables $y=X+h$, $G(h)$ is the
logarithm of the normalizing integral for \eqref{def:Vh}.  For any $v\in\mathbb{R}^d$,
\begin{align*}
 \inner{v}{D^2G(h)v}
 &=\Var_{\mathcal E_h}\!\left(\inner{\Sigma^{-1}v}{y}\right)
   -\inner{v}{\Sigma^{-1}v}\\
 &\leq\inner{\Sigma^{-1}v}{(\Sigma^{-1}-\gamma\one)^{-1}\Sigma^{-1}v}
   -\inner{v}{\Sigma^{-1}v}\\
 &=\inner{v}{\gamma(\one-\gamma \Sigma)^{-1}v},
\end{align*}
where the inequality is \eqref{eq:BL}.  The case $\Sigma\geq0$ follows by replacing $\Sigma$ with $\Sigma+\epsilon\one$ and sending
$\epsilon\downarrow0$.
\end{proof}

The following Poincar\'e estimate~\eqref{eq:Poincare-scalar} is a well-known consequence of the Brascamp--Lieb inequality~\cite{BL76} under the convexity criterion~\eqref{eq:strong-convexity}.  By the Bakry--\'Emery theorem \cite{BM85}, the latter also implies the logarithmic
Sobolev estimate~\eqref{eq:LSI}, which addresses the entropy,  $\Ent_{\mathcal E_{\bh}}[f^2] \coloneqq \mathcal E_{\bh}\left[f^2 \left( \log f^2 - \log \mathcal E_{\bh}f^2 \right)\right] $, under the law $ \mathcal E_{\bh} $. The 
exponential-moment bound then follows by the Herbst argument
\cite{Ledoux01}.

\begin{proposition}[Poincar\'e, LSI, and concentration]\label{lem:LSI}
Under the convexity condition~\eqref{eq:strong-convexity}, for continuously differentiable  $f:\mathbb{R}^d\to \mathbb{R} $:
\begin{align}
 \Var_{\mathcal E_{\bh}}(f)
 &\leq\frac1\kappa \ \mathcal E_{\bh}|\nabla f|_2^2,
 \label{eq:Poincare-scalar}\\
 \Ent_{\mathcal E_{\bh}}(f^2)
 &\leq\frac2\kappa \ \mathcal E_{\bh}|\nabla f|_2^2.
 \label{eq:LSI}
\end{align}
If $|\nabla f|_2^2\coloneqq \sum_{j=1}^d |\partial_j f |^2 \leq L_f^2$ almost everywhere, then for every
$\lambda\in\R$,
\begin{equation}\label{eq:Herbst}
 \mathcal E_{\bh}\exp\{\lambda(f-\mathcal E_{\bh}f)\}
 \leq\exp\left(\frac{\lambda^2L_f^2}{2\kappa}\right).
\end{equation}
\end{proposition}

\subsection{Time discretization and Gaussian linearization}
To apply the results of the previous subsection to the Gibbs measure, we first need to discretize the functional integral~\eqref{eq:path-Gibbs}. This can be done in many ways: Trotterization as in Appendix~\ref{app:Trotter} or any $ L^2(0,1) $-approximation of the paths. For convenience, we will use the square-wave pulses from~\cite{MW25}. For their definition, we 
fix $D\in\N$ and let $ I_\alpha^D\coloneqq [(\alpha-1)/2^D,\alpha/2^D) $, 
\[
 e_\alpha^D=\sqrt{2^D}\,\one_{I_\alpha^D},
 \qquad \alpha=1,\dots,2^D.
\]
The $\big(e_\alpha^D\big) $ form an orthonormal system of step functions at scale
$2^{-D}$.  Let $Q_D$ be the associated orthogonal projection in the Hilbert space $ L^2(0,1) $ and set
\begin{equation}\label{def:discrete-path}
 \xi_j^D(\alpha)\coloneqq \inner{e_\alpha^D}{\xi_j},
 \qquad \xi_j^D \coloneqq (\xi_j^D(1),\dots,\xi_j^D(2^D)).
\end{equation}
Discretizing the action in the functional integral yields the approximate partition function
\begin{align}
 Z_N^D(J,\bh) \coloneqq \int\exp\Bigg\{
 &\frac\beta2\sum_{\alpha=1}^{2^D}\sum_{j,k=1}^N
 J_{jk}\xi_j^D(\alpha)\xi_k^D(\alpha)
 +\beta 2^{-D/2}\sum_{\alpha=1}^{2^D}\sum_{j=1}^Nh_j\xi_j^D(\alpha)
 \Bigg\}\nu_N(\dd\bxi).
 \label{eq:discrete-Z}
\end{align}
The longitudinal term is exact because the constant function belongs to
the range of $ Q_D$.
\begin{proposition}[Removal of the discretization]\label{prop:discrete-limit}
For fixed $N$ and $J$, 
\begin{equation}\label{eq:pressureconv}
\lim_{D\to \infty} \log Z_N^D(J,\bh) = \log Z_N(J,\bh) , 
\end{equation}
the convergence is uniform in $ \bh $ on any compact subset, and its first two derivatives with respect
to $\bh$ converge locally uniformly in $\bh$ to the corresponding quantities.  Moreover,
\begin{equation}\label{eq:overlap-discrete-limit}
 \lim_{D\to \infty} \Big\| S_N - \frac1N\sum_{j=1}^N|Q_D\xi_j\rangle\langle Q_D\xi_j|
 \Big\|_{\HS} = 0
\end{equation}
for $\nu_N$-almost every path configuration and in all
bounded Gibbs moments.
\end{proposition}

\begin{proof}
Since $Q_D\xi_j\to\xi_j$ in $L^2(0,1)$,
\[
 \lim_{D\to \infty} \sum_{j,k=1}^N J_{jk}\inner{Q_D\xi_j}{Q_D\xi_k}
 =\sum_{j,k=1}^N J_{jk}\inner{\xi_j}{\xi_k}.
\]
The absolute value of either interaction is bounded by $N\norm J$, and the
longitudinal term is bounded uniformly on compact sets of $\bh$.  Dominated
convergence therefore applies to the partition function, which proves~\eqref{eq:pressureconv}, as well as to the first two
$\bh$-derivatives, whose additional factors are bounded.  Finally,
\[
 \norm{|Q_D\xi\rangle\langle Q_D\xi|-|\xi\rangle\langle\xi|}_{\HS}
 \leq2\norm{Q_D\xi-\xi}_2,
\]
which proves \eqref{eq:overlap-discrete-limit} and the moment convergence.
\end{proof}

Let $\bw=(w_{j,\alpha})$ be a centered Gaussian random vector on
$\R^N\otimes\R^{2^D}$ with covariance
\begin{equation}\label{eq:Gausslin0}
 \mathcal E[w_{j,\alpha}w_{k,\gamma}]
 =\frac1\beta J_{jk}\delta_{\alpha\gamma}
 =:\Sigma_{j\alpha,k\gamma}.
\end{equation}
Gaussian linearization allows us to write 
\begin{equation}\label{eq:Gausslin}
 Z_N^D(J,\bh)
 =\mathcal E\exp F_N^D(\bw+\widetilde\bh),
 \qquad \mbox{with} \; 
 (\widetilde\bh)_{j,\alpha}=2^{-D/2}h_j,
\end{equation}
and 
\begin{equation}\label{def:FD}
 F_N^D(\bw)
 =\sum_{j=1}^N\log\int
 \exp\left(\beta\sum_{\alpha=1}^{2^D}w_{j,\alpha} \ \xi^D(\alpha)\right)
 \nu_1(\dd\xi).
\end{equation}
Denoting by 
$
  \widehat w_j(s)
  :=\sum_{\alpha=1}^{2^D}w_{j,\alpha}e_\alpha^D(s)
  =2^{D/2}\sum_{\alpha=1}^{2^D}w_{j,\alpha}\one_{I_\alpha^D}(s) $ such that $   \sum_{\alpha=1}^{2^D}w_{j,\alpha}\xi_j^D(\alpha)
  =\langle \widehat w_j,\xi_j\rangle $, 
the second partial derivatives of this quantity are 
$$
\frac{\partial^2 F_N^D}{\partial w_{j,\alpha}  \ \partial w_{j',\alpha'} }(\mathbf{w})   =  \beta^2 \  \delta_{j,j'} \ \langle e_\alpha^D ,  \mu_{\widehat w_j} e_{\alpha'}^D \rangle  . $$
Using the Ding-Song-Sun inequality, Proposition~\ref{prop:DSS}, we  conclude that in the sense of quadratic forms:
\begin{equation}\label{eq:FD-Hessian}
D^2  F_N^D(\mathbf{w}) \leq \beta^2 \frac{\tanh(\beta b)}{\beta b} \mathbbm{1} = \gamma  \mathbbm{1} , \quad \mbox{with} \quad \gamma \coloneqq  \frac{\beta \tanh(\beta b)}{b} .
\end{equation}

\subsection{Proof of the Hessian bound}\label{sub:Hessian}

\begin{proof}[Proof of Theorem~\ref{thm:Hessian}]
The covariance in \eqref{eq:Gausslin0} is of the form
$ \Sigma=\beta^{-1}J\otimes\one_{2^D} $. 
Consequently, $\gamma\norm \Sigma=\norm{J} \tanh(\beta b)/b<1$, so we can apply 
Lemma~\ref{lem:gaussian-convolution} to \eqref{eq:Gausslin}.  Since the map
$\bh\mapsto\widetilde\bh$ is an isometry from $\R^N$ into the tensor product,
we obtain
\begin{align*}
 D_{\bh}^2\log Z_N^D(J,\bh)
 &\leq\gamma\bigl(\one-\gamma\beta^{-1}J\bigr)^{-1}\\
 &\leq\frac{\gamma}{1-\gamma\beta^{-1}\norm{J}}\one
 =\frac{\beta\tanh(\beta b)}{b-\norm{J}\tanh(\beta b)}\one ,
\end{align*}
where the second inequality is by the operator-monotonicity of the inverse. 
Proposition~\ref{prop:discrete-limit} allows us to take the limit $D\to\infty$ and proves
\eqref{eq:Hessian}.
\end{proof}

\subsection{Auxiliary results}

The proof of our concentration estimates is based on the following convenient form of a concentration result in Hilbert spaces due to Pinelis~\cite{Pin94}.

\begin{proposition}[Pinelis]\label{prop:Pinelis}
Let $X_1,\dots,X_N$ be independent centered random variables in a real Hilbert
space.  Then, for every $t>0$,
\begin{equation}\label{eq:Pinelis}
 \E\exp\left(t \ \Big\|\sum_{j=1}^NX_j\Big\|\right)
 \leq2\exp\left(\frac{t^2}{2}\sum_{j=1}^N
 \E\bigl[\norm{X_j}^2e^{t\norm{X_j}}\bigr]\right).
\end{equation}
\end{proposition}

\begin{proof}[Sketch of the proof of Thm.~3.2 in \cite{Pin94}]
For fixed vectors $x,y$, the function
$g(s)=\cosh(t\norm{x+sy})$ satisfies, wherever differentiable,
\[
 g''(s)\leq t^2\norm y^2e^{t\norm y}\cosh(t\norm x),
 \qquad 0\leq s\leq1.
\]
The inequality extends through the points $x+sy=0$ by approximation.  Conditioning
on $X_1,\dots,X_{N-1}$ and applying Taylor's formula in $s$ with
$x=\sum_{j<N}X_j$ and $y=X_N$, the conditional expectation of $g'(0)$ vanishes
because $\E X_N=0$.  Iteration and $1+u\leq e^u$ give the right-hand side of
\eqref{eq:Pinelis} with $\cosh$ in place of the exponential.  Finally, one uses
$e^a\leq2\cosh a$ for $a\geq0$.
\end{proof}

In applications of the Poincar\'e inequality, we will also need the following cross-covariance estimate.
\begin{proposition}[Covariance bound]\label{prop:covariance}
Let $\mathcal H$ and $\mathcal K$ be real Hilbert spaces, and let
$Y\in L^2(\Omega;\mathcal H)$ and $U\in L^2(\Omega;\mathcal K)$ be centered, Hilbert-space-valued random variables on a probability space $ \Omega $, and 
set
 $ \Sigma_Y\coloneqq\E|Y\rangle\langle Y| $, $  \Sigma_U\coloneqq\E|U\rangle\langle U| $. 
We define $T:\mathcal H\to\mathcal K$ by
$  T\varphi\coloneqq\E\bigl[U\,\inner{Y}{\varphi}\bigr] $. 
Then
\begin{equation}\label{eq:cross}
 TT^*\leq\norm{\Sigma_Y}\,Sigma_U,
 \qquad
 \tr TT^* = \tr T^* T
 \leq\norm{\Sigma_Y}\,\E\norm{U}_{\mathcal K}^2.
\end{equation}
\end{proposition}
\begin{proof}
For $v\in\mathcal K$, the Cauchy-Schwarz inequality implies  the norm-bound
\begin{align*}
 \norm{T^*v}_{\mathcal H}^2
 &=\sup_{\norm{\varphi}_{\mathcal H}=1}
   \left|\E\bigl[\langle U,v\rangle\langle Y,\varphi\rangle\bigr]\right|^2 \leq\sup_{\norm{\varphi}_{\mathcal H}=1}
   \E\langle U,v\rangle^2\,\E\langle Y,\varphi\rangle^2\\
 &\leq\norm{\Sigma_Y}\,\langle v,\Sigma_Uv\rangle.
\end{align*}
This proves the operator inequality. Taking its trace gives the Hilbert--Schmidt bound.
\end{proof}

\subsection{Proof of self-overlap concentration}\label{sec:self-proof}

\begin{proof}[Proof of Theorem~\ref{thm:concentration1} -- first part.] 
By replacing $J$ with $J+\eta\one$, which leaves the Gibbs state unchanged, and then sending $\eta\downarrow0$, it is enough to
prove the estimates below for $J>0$.

For the discretized system,  we abbreviate 
\begin{equation}\label{def:XD}
 X_j^D \coloneqq |\xi_j^D\rangle\langle\xi_j^D|,
 \qquad S_N^D=\frac1N\sum_{j=1}^NX_j^D.
\end{equation}
Conditioning in \eqref{eq:Gausslin} produces independent one-site measures
\begin{equation}\label{def:inner-average}
 \Av_{w_j}(\cdot)
 =\frac{\int (\cdot)e^{\beta\inner{w_j}{\xi^D}}\nu_1(\dd\xi)}
 {\int e^{\beta\inner{w_j}{\xi^D}}\nu_1(\dd\xi)},
\end{equation}
and, after the change of variables $\bw\mapsto\bw+\widetilde\bh$, an
outer probability measure $\mathcal E_{\bh}^D$ proportional to
\begin{equation}\label{eq:outer-measure}
 \exp\left[-\frac12\inner{\bw-\widetilde\bh}{\Sigma^{-1}(\bw-\widetilde\bh)}
 +F_N^D(\bw)\right]\dd\bw.
\end{equation}
By~\eqref{eq:FD-Hessian}, its convexity constant is
\begin{equation}\label{def:alpha}
 \kappa=\frac{\beta}{ \norm{J} }-\frac{\beta\tanh(\beta b)}{b}
 =\frac{\beta\bigl(b-\norm{J}\tanh(\beta b)\bigr)}{b\norm{J}} . 
\end{equation}

We will be interested in the inner average of the discretized self-overlap, which defines the following $\R^{2^D\times 2^D }$-valued function of $ \bw $:
\begin{equation}\label{def:AN}
 A_N^D(\bw)\coloneqq \Av_{\bw}(S_N^D)
 =\frac1N\sum_{j=1}^N\Av_{w_j}(X_j^D) .  
\end{equation}
For an application of Proposition~\ref{prop:covariance}, we identify $\mathbb R^{2^D}$ with $\mathcal H_D\coloneqq\Ran Q_D$ and set
$\mathcal K_D\coloneqq\mathcal S_2(\mathcal H_D)$ for the Hilbert space of Hilbert-Schmidt operators on $\mathcal H_D$. For fixed $w_j$, we define the
cross-covariance operator $T_j^D(w_j):\mathcal H_D\to\mathcal K_D$ by
\[
 T_j^D(w_j)\varphi
 \coloneqq\Av_{w_j}\!\left[
  \bigl(X_j^D-\Av_{w_j}X_j^D\bigr)
  \inner{\xi_j^D-\Av_{w_j}\xi_j^D}{\varphi}
 \right].
\]
Its input covariance is $Q_D\mu_{\widehat w_j}Q_D$, for which by
Proposition~\ref{prop:DSS},
\[
 \norm{Q_D\mu_{\widehat w_j}Q_D}
 \leq\norm{\mu_0}
 \leq\frac{\tanh(\beta b)}{\beta b}.
\]
Moreover, $\| X_j^D\|_{\HS}\leq1$, and hence 
$  \Av_{w_j}\norm{X_j^D-\Av_{w_j}X_j^D}_{\HS}^2\leq1 $. 
The cross-covariance bound~\eqref{eq:cross} therefore gives, uniformly in $w_j$,
\[
  \norm{T_j^D(w_j)}_{\HS(\mathcal H_D,\mathcal K_D)}^2 =  \tr  T_j^D(w_j)^* T_j^D(w_j)  
 \leq\norm{\mu_0}
 \leq\frac{\tanh(\beta b)}{\beta b}.
\]
Using the scalar Poincar\'e inequality componentwise in an orthonormal basis of
$\mathcal K_D$, we consequently obtain
\begin{align}
\mathcal E_{\bh}^D\!\left[
 \norm{A_N^D-\mathcal E_{\bh}^D[A_N^D]}_{\HS}^2\right]
 &\leq\frac{\beta^2}{\kappa N^2}\sum_{j=1}^N
 \mathcal E_{\bh}^D\!\left[
  \norm{T_j^D(w_j)}_{\HS(\mathcal H_D,\mathcal K_D)}^2\right]\notag\\
 &\leq\frac{\beta^2}{\kappa N}\norm{\mu_0}
 \leq\frac{\beta}{\kappa N}\frac{\tanh(\beta b)}{b}
 =\frac{r(\beta,b;\norm{J})}{N}.
 \label{eq:outer-variance}
\end{align}
Here, we recall the parameter $r(\beta,b;\norm{J})$, which we abbreviate in the following by $r$, from the statement of Theorem~\ref{thm:concentration1}.

For the remainder of the proof, we investigate
\[
  f_N^D(\bw)\coloneqq\norm{U_N^D(\bw)}_{\HS},
 \qquad. U_N^D(\bw)\coloneqq A_N^D(\bw)-\mathcal E_{\bh}^D[A_N^D].
\]
At points where $f_N^D(\bw)>0$, differentiation gives
\[
 D_{w_j}f_N^D
 =\frac{\beta}{Nf_N^D}(T_j^D)^*U_N^D.
\]
The cross-covariance bound~\eqref{eq:cross} and $\| X_j^D \|_{\HS}\leq1$ imply
\[
 \norm{(T_j^D)^*U_N^D}_2^2
 \leq\norm{\mu_0}\,
 \Var_{\Av_{w_j}}\!\left(\langle U_N^D,X_j^D\rangle_{\HS}\right)
 \leq\norm{\mu_0}\,\norm{U_N^D}_{\HS}^2 , 
\]
such that 
\begin{align}\label{eq:diffLSI}
\sum_{j,\alpha}\left|\frac{\partial f_N^D}{\partial w_{j,\alpha}}(\bw)\right|^2
 &\leq\frac{\beta^2}{N}\norm{\mu_0}
 \leq\frac{\beta}{N}\frac{\tanh(\beta b)}{b}.
\end{align}
The estimate extends to the points where $f_N^D=0$ by almost-everywhere differentiability of the norm.
Hence, by Herbst's estimate~\eqref{eq:Herbst} for any $ t \geq 0 $
\begin{equation}\label{eq:outer-mgf}
 \mathcal E_{\bh}^D\exp\left( tNf_N^D\right) 
 \leq\exp\left(t\sqrt{Nr}+\frac{Nt^2r}{2}\right).
\end{equation}

Conditioned on $\bw$, the centered variables
$Y_j=X_j^D-\Av_{w_j}(X_j^D)$ are independent,
$\norm{Y_j}_{\HS}\leq2$, and
$\Av_{w_j}\norm{Y_j}_{\HS}^2\leq1$.  Pinelis bound, Proposition~\ref{prop:Pinelis}, gives
\begin{equation}\label{eq:inner-mgf}
 \Av_{\bw}\left[ \exp\left(tN\norm{S_N^D-A_N^D}_{\HS}\right)\right]
 \leq2\exp\left(\frac{Nt^2e^{2t}}2\right).
\end{equation}
The bound is uniform in $\bw$.  Combining \eqref{eq:outer-mgf} and
\eqref{eq:inner-mgf}, using the triangle inequality, and then taking
$D\to\infty$ yields, for $0<t\leq\tfrac12$,
\begin{equation}\label{eq:self-mgf}
 \left\langle e^{tN\norm{S_N-\langle S_N\rangle}_{\HS}}\right\rangle_{J,\bh}
 \leq2\exp\left(t\sqrt{Nr}+\frac{Nt^2}{2}(e+r)\right).
\end{equation}
The exponential Markov inequality  and the choice
$t=\min\{\epsilon/(e+r),1/2\}$ prove \eqref{eq:soconc}.

For the second moment, we use the conditional variance decomposition and
\eqref{eq:outer-variance}.  Since we also have $
 \Av_\bw\norm{S_N^D-A_N^D}_{\HS}^2\leq N^{-1} $, we conclude 
\[
 \left\langle\norm{S_N^D-\langle S_N^D\rangle}_{\HS}^2\right\rangle_{J,\bh}
 \leq\frac{1+r}{N}.
\]
Letting $D\to\infty$ proves \eqref{eq:self-variance}.
\end{proof}

It remains to establish  \eqref{eq:sharp-self-bound}.  At $\bh=\mathbf{0} $ the function
\begin{equation}\label{def:CN}
 C_N(u)=\langle S_N(u,0)\rangle_{J,0}
 =\frac1N\sum_{j=1}^N\langle S_j^z(u)S_j^z\rangle_{J,0}
\end{equation}
is non-negative, time-translation invariant, and has the KMS spectral
representation \cite[Ch.~5]{BR1997}
\begin{equation}\label{eq:KMSrep}
 C_N(u)=\int_0^\infty K_x(u)\,\rho_N(\dd x),
 \qquad
 K_x(u)=\frac{\cosh(x(1/2-u))}{\cosh(x/2)},
\end{equation}
where $\rho_N$ is a probability measure.

\begin{lemma}[Square of a KMS autocorrelation]\label{lem:autocorr-square}
Let $C(u)=\int K_x(u)\rho(\dd x)$ with a finite positive measure satisfying
$\rho([0,\infty))\leq1$.  If $q=\int_0^1C(u)\dd u$, then 
\begin{equation}\label{eq:autocorr-square}
 \int_0^1C(u)^2\dd u\leq\frac q2+q^2.
\end{equation}
\end{lemma}

\begin{proof}
Abbreviating $c(x)\coloneqq \int_0^1K_x(u)\dd u=2\tanh(x/2)/x$, with $c(0)=1$, a computation
 gives
\[
 \int_0^1K_x(u)^2\dd u
 =\frac12c(x)+\frac12\operatorname{sech}^2(x/2)
 \leq\frac12c(x)+c(x)^2.
\]
The Cauchy--Schwarz inequality and the identity
\[
 \left(\frac{a+b}{4}+ab\right)^2
 -\left(\frac a2+a^2\right)\left(\frac b2+b^2\right)
 =\frac{(a-b)^2}{16}
\]
show that
$\int_0^1K_x(u) K_y(u) \dd u\leq(c(x)+c(y))/4+c(x)c(y)$.  Integrating in $x,y$ against the measure 
$\rho$ proves \eqref{eq:autocorr-square}.
\end{proof}

We are now ready to complete the proof of Theorem~\ref{thm:concentration1}.
\begin{proof}[Proof of Theorem~\ref{thm:concentration1} -- second part.] 
By \eqref{eq:Duhamel-def}, \eqref{eq:nohmag}, and Theorem~\ref{thm:Hessian},
\begin{equation}\label{eq:CN-integral}
 \int_0^1C_N(u)\dd u
 =\frac1N\sum_{j=1}^N(S_j^z,S_j^z)_{J,0}
 \leq m(\beta,b;\norm{J}).
\end{equation}
Applying Lemma~\ref{lem:autocorr-square} proves
\eqref{eq:sharp-self-bound} and completes the proof of
Theorem~\ref{thm:concentration1}.
\end{proof}

\subsection{Proof of replica-overlap concentration}\label{sec:replica-proof}

\begin{proof}[Proof of  Theorem~\ref{thm:repconc}]
We spell out the argument after discretization and then let $D\to\infty$. For brevity, write
$r=r(\beta,b;\norm{J})$ and $\mathfrak m=\mathfrak m(\beta,b;\norm{J})$ and set
\[
 \bar\xi_j^D\coloneqq\langle\xi_j^D\rangle_{J,\bh}^{(D)},
 \qquad
 u_j(w_j)\coloneqq\Av_{w_j}(\xi_j^D)-\bar\xi_j^D,
 \qquad
 u_{j,\alpha}(w_j)\coloneqq\langle e_\alpha^D,u_j(w_j)\rangle
\]
with the Gibbs state $\langle\cdot\rangle_{J,\bh}^{(D)}$ of the discretized model.
Given two independent outer fields $\bw=(\bw^{(1)},\bw^{(2)})$, the conditional mean of
the centered replica overlap is defined by the kernel
\begin{equation}\label{def:replica-conditional-mean}
 B_N^D(\bw)(\alpha,\alpha')
 \coloneqq\frac1N\sum_{j=1}^N
 u_{j,\alpha}(w_j^{(1)})u_{j,\alpha'}(w_j^{(2)}).
\end{equation}
Conditioned on the two outer fields, we set
\[
 Z_j\coloneqq
 \bigl(\xi_j^{D,(1)}-\bar\xi_j^D\bigr)\otimes
 \bigl(\xi_j^{D,(2)}-\bar\xi_j^D\bigr)
 -u_j(w_j^{(1)})\otimes u_j(w_j^{(2)}).
\]
Then the $(Z_j)$ are independent, centered $\mathcal K_D$-valued random variables and
\[
 C_N^{D,(1,2)}-B_N^D=\frac1N\sum_{j=1}^N Z_j.
\]
Since $\|\xi_j^D-\bar\xi_j^D\|_2\leq2$, one has $
 \norm{Z_j}_{\HS}\leq8 $ and 
$
 \Av_{\bw}^{\otimes}\norm{Z_j}_{\HS}^2\leq16 $. 
Proposition~\ref{prop:Pinelis} therefore gives
\begin{equation}\label{eq:replica-inner-mgf}
 \Av^\otimes_{\bw}\left[
 \exp\left(tN\norm{C_N^{D,(1,2)}-B_N^D}_{\HS}\right)\right]
 \leq2e^{8Nt^2e^{8t}}.
\end{equation}

The product outer measure $\mathcal E_{\bh}^{D,\otimes}$ satisfies the tensorized logarithmic Sobolev inequality with the same constant $\kappa$ from~\eqref{def:alpha}. We set $f_N^D=\norm{B_N^D}_{\HS}$ and abbreviate
$M_j^{(a)}\coloneqq Q_D\mu_{\widehat w_j^{(a)}}Q_D $ with $ a\in\{1,2\}$.
Then
\[
 D_{w_j^{(a)}}u_j(w_j^{(a)})=\beta M_j^{(a)},
 \qquad \norm{M_j^{(a)}}\leq\norm{\mu_0},
\]
where we used Proposition~\ref{prop:DSS} again. 
Viewing $B_N^D$ as an operator from the second copy of $\mathcal H_D$ to the first,
differentiation in the first outer field gives, whenever $f_N^D>0$,
\begin{align}
\sum_{j,\alpha}\left|\frac{\partial f_N^D}{\partial w_{j,\alpha}^{(1)}}\right|^2
 &=\frac{\beta^2}{N^2(f_N^D)^2}\sum_{j=1}^N
   \norm{M_j^{(1)}B_N^D u_j(w_j^{(2)})}_2^2\notag\\
 &\leq\frac{4\beta^2}{N}\norm{\mu_0}^2
 \leq\frac{4\beta}{N}\frac{\tanh(\beta b)}{b}.
 \label{eq:replica-gradient-first}
\end{align}
Here we used $\norm{B_N^D u}_2\leq\norm{B_N^D}_{\HS}\norm{u}_2$,
$\|u_j(w_j^{(2)})\|_2\leq2$, and $\norm{\mu_0}\leq1$. The same estimate holds
for the second copy, and therefore
\begin{equation}\label{eq:replica-gradient}
 \norm{\nabla f_N^D}_2^2\leq\frac{8\beta}{N}\frac{\tanh(\beta b)}{b}.
\end{equation}
Furthermore, Jensen's inequality, the removal of the discretization, and
\eqref{eq:orderp} yield
\begin{equation}\label{eq:replica-mean-outer}
 \limsup_{D\to\infty}\mathcal E_{\bh}^{D,\otimes}[f_N^D]
 \leq\left(\left\langle\norm{C_N^{(1,2)}}_{\HS}^2\right\rangle_{J,\bh}^{\otimes}\right)^{1/2}
 \leq\frac{\mathfrak m(\beta,b;\norm{J})}{\sqrt N}.
\end{equation}
The Herbst bound~\eqref{eq:Herbst} consequently guarantees that 
\begin{equation}\label{eq:replica-outer-mgf}
 \limsup_{D\to\infty}\mathcal E_{\bh}^{D,\otimes}e^{tNf_N^D}
 \leq\exp\left(t\mathfrak m\sqrt N+4rNt^2\right).
\end{equation}
Combining \eqref{eq:replica-inner-mgf} and \eqref{eq:replica-outer-mgf}, for
$0<t\leq1/8$, gives
\begin{equation}\label{eq:replica-mgf}
 \left\langle e^{tN\norm{C_N^{(1,2)}}_{\HS}}\right\rangle_{J,\mathbf h}^{\otimes}
 \leq2\exp\left(t\mathfrak m\sqrt N+\frac{Nt^2}{2}(16e+8r)\right).
\end{equation}
Markov's inequality and $t=\min\{\epsilon/\sigma_{\mathrm r}^2,1/8\}$ imply
\eqref{eq:reconc}. This completes the proof of Theorem~\ref{thm:repconc}.
\end{proof}

\appendix

\section{Trotter approximation and the loop representation}\label{app:Trotter}

Let $v:\{-1,1\}\times[0,1]\to\R$ be bounded and Lipschitz in time, and set
$V(s)=v(S^z;s)$.  The time-ordered evolution generated by $V(s)+bS^x$ is
\begin{equation}\label{eq:Dyson}
 T_{t,s}=\one+\sum_{n\geq1}
 \int_{s\leq t_1\leq\cdots\leq t_n\leq t}
 \bigl(V(t_n)+bS^x\bigr) \cdots  \bigl(V(t_1)+bS^x\bigr) \,
 \dd t_1\cdots\dd t_n.
\end{equation}
It solves $\partial_tT_{t,s}=(V(t)+bS^x)T_{t,s}$ and $T_{s,s}=\one$.
The Lie--Trotter product formula \cite{SIC13,ZNI24} gives, in operator norm,
\begin{equation}\label{eq:Trotter-product}
 T_{1,0}=\lim_{M\to\infty}
 \prod_{\alpha=M}^{1}
 \exp\left(\frac1M V(\alpha/M)\right)
 \exp\left(\frac bM S^x\right),
\end{equation}
where the product is ordered from right to left in increasing time.
In the $S^z$-basis,
\begin{equation}\label{eq:Sx-matrix-elements}
 \langle\sigma|e^{aS^x}|\tau\rangle
 =\cosh(a)\one_{\{\sigma=\tau\}}+\sinh(a)\one_{\{\sigma\neq\tau\}}
 =N_a e^{K_a\sigma\tau},
\end{equation}
where
\[
 N_a=\sqrt{\frac{\sinh(2a)}2},
 \qquad K_a=-\frac12\log\tanh(a).
\]
Consequently,
\begin{align}
 \Tr T_{1,0}
 =\lim_{M\to\infty}N_{b/M}^{M}
 \sum_{\substack{\sigma_1,\dots,\sigma_M=\pm1\\\sigma_0=\sigma_M}}
 \exp\Bigg(&K_{b/M}\sum_{\alpha=1}^M\sigma_{\alpha-1}\sigma_\alpha +\frac1M\sum_{\alpha=1}^Mv(\sigma_\alpha;\alpha/M)\Bigg).
 \label{eq:chain-representation}
\end{align}
The chain is ferromagnetic because $K_{b/M}>0$.  Since the $v=0$ expression
is $2\cosh b$, positivity of all matrix elements gives
\begin{equation}\label{eq:trace-bound}
 e^{-\norm v_\infty}
 \leq\frac{\Tr T_{1,0}}{2\cosh b}
 \leq e^{\norm v_\infty}.
\end{equation}
Insertions of $S^z$ converge to insertions of the corresponding chain spins;
for example, for $0\leq s\leq t\leq1$,
\begin{align}
 \frac{\Tr(T_{1,t}S^zT_{t,0})}{\Tr T_{1,0}}
 &=\lim_{M\to\infty}\langle\sigma_{\lfloor Mt\rfloor}\rangle_M,
 \label{eq:Trotter-one-point}\\
 \frac{\Tr(T_{1,t}S^zT_{t,s}S^zT_{s,0})}{\Tr T_{1,0}}
 &=\lim_{M\to\infty}\langle
 \sigma_{\lfloor Mt\rfloor}\sigma_{\lfloor Ms\rfloor}\rangle_M.
 \label{eq:Trotter-two-point}
\end{align}
These limits justify the use of FKG and the Ding--Song--Sun inequality in the
proof of Proposition~\ref{prop:DSS}.

\section{Concentration of the Gibbs-averaged self-overlap}\label{app:sac}
The following complements Theorem~\ref{thm:concSK} for the QSK.
\begin{corollary}
Assume~\eqref{eq:QSK-condition}, let $\bh=\mathbf 0$, and fix $\delta>0$ such that
\[
 b>(4+2\delta)\tanh(\beta b).
\]
Then there is $L=L(\beta,b,\delta)<\infty$ such that, for every $\varepsilon>0$ and every $N$,
\begin{equation}\label{eq:concGauss}
 \Pp\!\left(
  \norm{\langle S_N\rangle_{G,\mathbf 0}-\E\langle S_N\rangle_{G,\mathbf 0}}_{\HS}
  >\varepsilon+\frac{L}{\sqrt N}
 \right)
 \leq\exp\!\left(-\frac{N\varepsilon^2}{L^2}\right)
 +\Pp(\Omega_{N,\delta}^{\mathrm c}).
\end{equation}
\end{corollary}
Since an admissible $\delta$ may be fixed as a function of $(\beta,b)$, the constant may equivalently be denoted by $L(\beta,b)$.
\begin{proof}
We write the independent Gaussian coordinates collectively as $g=(g_{jk})_{1\leq j <k \leq N}$ and study
\[
 F_N(g)\coloneqq\langle S_N\rangle_{G(g),\mathbf 0}
 \in\mathcal K\coloneqq\mathcal S_2(L^2(0,1)),
 \qquad
 \Xi_N\coloneqq S_N-\langle S_N\rangle_{G,\mathbf 0}.
\]
We also abbreviate $
 r_\delta\coloneqq r(\beta,b;4+2\delta) $. 
On $\Omega_{N,\delta}$, differentiation of all Hilbert--Schmidt components and
the normalization~\eqref{def:GOE} give
\begin{align}
 \sum_{j,k}\norm{\frac{\partial F_N}{\partial g_{jk}}}_{\HS}^2
 &\leq\beta^2N\left\langle
  \norm{\Xi_N^{(1)}}_{\HS}\norm{\Xi_N^{(2)}}_{\HS}
  \norm{R_N^{(1,2)}}_{\HS}^2
 \right\rangle_{G,\mathbf 0}^{\otimes}\notag\\
 &\leq\beta^2N\left\langle\norm{\Xi_N}_{\HS}^2\right\rangle_{G,\mathbf 0}
 \left(\left\langle\norm{R_N^{(1,2)}}_{\HS}^4\right\rangle_{G,\mathbf 0}^{\otimes}\right)^{1/2}\notag\\
 &\leq\frac{\beta^2(1+r_\delta)\,C_4(\beta,b;4+2\delta)^{1/2}}{N}
 \eqqcolon\frac{L_0^2}{N}.
 \label{eq:appendix-gradient-good}
\end{align}
At zero longitudinal field the centered and uncentered replica overlaps coincide so that the last line is implied by~\eqref{eq:self-variance} and~\eqref{eq:replmom}. 

The set
$ \Gamma_{N,\delta}\coloneqq\{g:\norm{G(g)}\leq2+\delta\} $ 
is closed and convex in the Euclidean space of independent Gaussian coordinates. Let
$\Pi_{N,\delta}$ be its metric projection and define
\[
 \widehat F_N(g)\coloneqq F_N(\Pi_{N,\delta}g).
\]
The metric projection is nonexpansive and differentiable almost everywhere, with
$\norm{D\Pi_{N,\delta}}\leq1$. Hence~\eqref{eq:appendix-gradient-good} and the chain rule imply the global bound almost everywhere:
\begin{equation}\label{eq:appendix-gradient-global}
 \sum_{j,k}\norm{\frac{\partial\widehat F_N}{\partial g_{jk}}}_{\HS}^2
 \leq\frac{L_0^2}{N} .
\end{equation}
Moreover, $\widehat F_N=F_N$ on $\Omega_{N,\delta}$ and
$\norm{F_N}_{\HS},\|\widehat F_N\|_{\HS}\leq1$.

Applying the Gaussian Poincar\'e inequality componentwise in an orthonormal basis of $\mathcal K$ gives
\begin{equation}\label{eq:appendix-poincare}
 \E\norm{\widehat F_N-\E\widehat F_N}_{\HS}^2\leq\frac{L_0^2}{N}.
\end{equation}
The scalar function $
 g\longmapsto\norm{\widehat F_N(g)-\E\widehat F_N}_{\HS} $
is $L_0/\sqrt N$-Lipschitz by~\eqref{eq:appendix-gradient-global}. Gaussian concentration and
\eqref{eq:appendix-poincare} therefore yield the bound
\begin{equation}\label{eq:appendix-concentration-projected}
 \Pp\left(
  \norm{\widehat F_N-\E\widehat F_N}_{\HS}
  >\varepsilon+\frac{L_0}{\sqrt N}
 \right)
 \leq\exp\left(-\frac{N\varepsilon^2}{2L_0^2}\right).
\end{equation}

Since $F_N=\widehat F_N$ on $\Omega_{N,\delta}$ and both maps have Hilbert-Schmidt norm at most one,
$
 \norm{\E F_N-\E\widehat F_N}_{\HS}
 \leq 2  \ \Pp(\Omega_{N,\delta}^{\mathrm c}) 
$, 
such that
\begin{equation}
 \Pp\left(
  \norm{F_N-\E F_N}_{\HS}
  >\varepsilon+\frac{L_0}{\sqrt N}
   +2 \ \Pp(\Omega_{N,\delta}^{\mathrm c})
 \right) \leq
 \exp\!\left(-\frac{N\varepsilon^2}{2L_0^2}\right)
 +\Pp(\Omega_{N,\delta}^{\mathrm c}).
 \label{eq:appendix-comparison}
\end{equation}
By the exponential tail estimate for $\Omega_{N,\delta}^{c}$ after~\eqref{def:omega}, we also have $
 a_\delta\coloneqq\sup_{N\geq1}\sqrt N\,\Pp(\Omega_{N,\delta}^{\mathrm c})<\infty $. 
The proof is completed by choosing $
 L\geq\max\{\sqrt2L_0,L_0+2a_\delta\} $, 
which turns~\eqref{eq:appendix-comparison} into~\eqref{eq:concGauss}.
\end{proof}


\paragraph{Acknowledgements.}
The authors acknowledge the use of AI tools for parts of the manuscript. All mathematical arguments and proofs in the final
manuscript were written by the authors.\\
This work was funded by the Deutsche Forschungsgemeinschaft (DFG, German
Research Foundation), grant 558731723 (CM), and under Germany's Excellence
Strategy, EXC-2111--390814868 (SW).

\paragraph{Conflicts of interest.}
The authors declare no competing interests relevant to this article.



\end{document}